\documentclass[runningheads]{llncs}

\usepackage{booktabs} 
\usepackage{amsfonts,amssymb,amsmath,latexsym,ae,aecompl,bbm,algorithm}
\usepackage{graphicx}
\usepackage[noend]{algpseudocode}
\usepackage{color}
\usepackage[misc,geometry]{ifsym}

\newcommand{\blue}[1]{\textcolor{blue}{#1}}
\newcommand{\magenta}[1]{\textcolor{magenta}{#1}}
\def\ShowComment{True}

\ifdefined\ShowComment
\def\anisur#1{\marginpar{$\leftarrow$\fbox{A}}\footnote{$\Rightarrow$~{\sf #1 \magenta{--Anisur}}}}
\else
\def\anisur#1{}
\fi

\ifdefined\ShowComment
\def\billy#1{\marginpar{$\leftarrow$\fbox{B}}\footnote{$\Rightarrow$~{\sf #1 \blue{--Billy}}}}
\else
\def\billy#1{}
\fi

\begin{document}
\title{Dispersion of Mobile Robots: The Power of Randomness}

\author{Anisur Rahaman Molla\inst{1}\thanks{Research supported in part by DST Inspire Faculty research grant DST/INSPIRE/04/2015/002801. ORCID ID: 0000-0002-1537-3462}\Letter
\and William K. Moses Jr.\inst{2}\thanks{Research supported in part by a grant of his postdoctoral fellowship hosts from the Israeli Ministry of Science. ORCID ID: 0000-0002-4533-7593}
}

\authorrunning{A.\,R. Molla and W.\,K. Moses Jr.}

\institute{Cryptology and Security Research Unit, Indian Statistical Institute, Kolkata 700108, India. \email{molla@isical.ac.in}
\and Faculty of Industrial Engineering and Management, Technion - Israel Institute of Technology, Haifa, Israel. \email{wkmjr3@gmail.com}
}

\maketitle

\begin{abstract}

We consider cooperation among insects, modeled as cooperation between mobile robots on a graph. Within this setting, we consider the problem of mobile robot dispersion on graphs. The study of mobile robots on a graph is an interesting paradigm with many interesting problems and applications. The problem of dispersion in this context, introduced by Augustine and Moses Jr. \cite{AM18}, asks that $n$ robots, initially placed arbitrarily on an $n$ node graph, work together to quickly reach a configuration with exactly one robot at each node. Previous work on this problem has looked at the trade-off between the time to achieve dispersion and the amount of memory required by each robot. However, the trade-off was analyzed for {\em deterministic algorithms} and the minimum memory required to achieve dispersion was found to be $\Omega(\log n)$ bits at each robot. In this paper, we show that by harnessing the power of {\em randomness}, one can achieve dispersion with $O(\log \Delta)$ bits of memory at each robot, where $\Delta$ is the maximum degree of the graph. Furthermore, we show a matching lower bound of $\Omega(\log \Delta)$ bits for any {\em randomized algorithm} to solve dispersion.

We further extend the problem to a general $k$-dispersion problem where $k> n$ robots need to disperse over $n$ nodes such that at most $\lceil k/n \rceil$ robots are at each node in the final configuration. 

\end{abstract}

\vspace{-0.3cm}
\section{Introduction}
\label{sec:intro}

\vspace{-0.2cm}
\subsection{Background \& Motivation}
The mobile robots paradigm has been used to study many types of systems, including those where simple insects cooperate with each other to accomplish some goal. These robots typically need to work together to solve some common problem such as shape formation or exploration of the environment or gathering at some common point. One of the primary motivations of this type of research is to understand how to use resource-limited robots to achieve some large task in a distributed manner.

Typically the environment that serves as a backdrop to these problems is a either a finite plane or a connected graph. However, a graph can just be thought of as a discretization of the space of a finite plane or in fact three dimensional space. Thus, using graphs as an environment allows, in some sense, for a more general study of a given problem. 

The problem of dispersion on graphs was recently introduced by Augustine and Moses Jr.~\cite{AM18}. The initial version of the problem asks that $n$ robots that are initially arbitrarily placed on a graph should work together to reach a final configuration such that there is exactly one robot on each node. We study this problem and the more general version of it where $k$ robots (for any $k$) are initially arbitrarily placed and must reach a configurations such that at most $\lceil k/n \rceil$ robots are present on any given node. The study of dispersion is interesting and has practical applications to any problem where the cost of several robots sharing the same resource (node) far outweighs the cost of a robot finding a new resources (moving on the graph). One such example is when multiple electric cars must find a recharge station in an area where recharge stations are located close by. The time to charge the vehicle may be in the order of hours while the time to find another station would be in the order of minutes. Further, if the vehicles are ``smart" and communicate with each other to exchange information about what stations are free or not, this problem is exactly modeled as dispersion. The study of dispersion is also interesting as it relates to the related problems of scattering on graphs, multi-robot exploration, and load balancing on graphs. Scattering on graphs asks that $k\leq n$ spread themselves out in an equidistance manner on symmetric graphs like rings or grids. This is just dispersion with an extra constraint of equi-spacing. Multi-robot exploration asks that $k$ robots starting at the same node work together to visit each node of the graph as quickly as possible. It is clear that any solution to dispersion solves this problem. Finally, load balancing on graphs asks that nodes send and receive loads and evenly distribute these loads among themselves. Dispersion can be seen as flipping this model by having the loads (i.e., robots) move around and distribute themselves evenly among the nodes. The techniques used to solve load balancing in graphs are quite different from those used to solve problems in mobile robots and by studying dispersion, our hope is to build a bridge between the two areas for cross-pollination of ideas and techniques.

As mentioned earlier, one of the key aspects of mobile robots is that we are solving large tasks in a distributed manner with \textit{resource limited mobile robots}. In previous work, the study of memory of robots and time to achieve dispersion was of great interest. This paper furthers that study and shows that the introduction of randomness in a novel way allows robots to achieve dispersion using much less memory than previously shown.

\vspace{-.3cm}
\subsection{Our Results}
Throughout this paper, we study the trade-off between memory required by robots and the time it takes to achieve dispersion of $n$ robots on different types of graphs with $n$ nodes and $m$ edges. We denote the diameter of these graphs by $D$ and maximum degree of any node of the graph by $\Delta$. We present algorithms for increasingly general types of graphs that utilize randomness to allow robots, with typically $O(\log \Delta)$ bits of memory\footnote{Note that all $\log$'s that appear in this paper are to the base $2$.}, to achieve dispersion. This is a substantial improvement over past algorithms which, while deterministic, required robots to have $\Omega(\log n)$ bits of memory each\footnote{Notice that our algorithms require only $O(1)$ bits memory at each robot on paths, rings, grids and any constant degree graphs, whereas the previous deterministic algorithms require $O(\log n)$ bits.}. We also show a lower bound on the memory requirement that any randomized algorithm requires $\Omega(\log \Delta)$ bits to achieve dispersion, assuming all robots have the same amount of memory.

When we consider a \textit{rooted graph} of a certain type, it implies that the topology of the graph is of that type and all robots start at one node, called the root, of the graph. We initially describe our algorithms for dispersion of $n$ robots on $n$ node graphs and subsequently generalize them to dispersion of any $k$ robots on $n$ node graphs. We assume that robots do not know the values of $n$, $m$, $k$, $\Delta$, or $D$. However, in several instances, our algorithms require robots to have memory proportionate to either parameter $\Delta$ or $D$. This means that whatever memory supplied to the robot should be enough to satisfy the requirement, but the explicit knowledge of the parameter itself is not needed. 
Our upper bound results for dispersion of $n$ robots on $n$ node graphs are summarized in Table~\ref{table:results}.

We first describe a primitive, \emph{Local-Leader-Election}, that can be used by robots with access to randomness to choose one robot to settle down at a given node. This  allows us to side-step the requirement of $\Omega(\log (k/n))$ bits of memory required by each robot for a unique label if we want robots to deterministically choose a robot to settle down at each node.

We then proceed to show how a simple algorithm for rooted rings, \emph{Rooted-Ring}, that requires robots to have $O(\log \Delta)$ bits of memory and achieves dispersion in $O(n)$ rounds. This serves as a warm-up and allows readers to internalize the way we use \emph{Local-Leader-Election} and how these sorts of algorithms (with reduced memory) need to operate.

We then develop the algorithm \emph{Rooted-Tree} for rooted trees that requires robots to have $O(\log \Delta)$ bits of memory and achieves dispersion in $O(n)$. This algorithm contains the key ideas for our algorithm on general rooted graphs.

We then present two algorithms to achieve dispersion on rooted graphs, \emph{Rooted-Graph-LogDelta-LogD} and \emph{Rooted-Graph-Delta}, which require robots to have $O(\max \lbrace \log \Delta, \log D \rbrace)$ and $O(\Delta)$ bits of memory respectively and both algorithms achieve dispersion in $O(m)$ rounds. Both algorithms are extensions of the \emph{Rooted-Tree} algorithm and use different amounts of memory to handle the issue of dealing with cycles that may arise in the graph. We provide these two algorithms with contrasting memory requirements so that if an algorithm designer a priori knows which of the two memory requirements is less, they can program robots to use that algorithm.

Finally, we present an algorithm \emph{Arbitrary-Graph}, which works on arbitrary graphs, and allows robots with $O(\log \Delta)$ bits of memory to achieve dispersion in the cover time of the graph with high probability. The algorithm is a Las Vegas type randomized algorithm in that robots will eventually achieve dispersion, but the exact running time is not fixed and but bounded with high probability. The ``cost" of having robots use less memory is that we require robots to stay active after they settle down. Namely, each robot runs the algorithm until it settles down and then must stay active to inform other robots that come to the node that the node is settled. In this sense, the algorithm is non-terminating. Contrast this against other algorithms which allow robots to settle down and need not be active after certain conditions are met.

After presenting the above algorithms and analyzing them for dispersion of $n$ robots on $n$ nodes, we then show how to generalize them to achieve dispersion of $k$ robots on $n$ nodes, where $k$ can be any positive integer value.

We then present our lower bound of $\Omega(\log \Delta)$ bits of memory needed by each robot to achieve dispersion when randomness is allowed.

\begin{table*}[ht]
	\caption{Upper bound results of dispersion of $n$ robots on an $n$ node graph (with $\Delta$ maximum degree and diameter $D$) for different types of graphs along with the memory requirement of robot.} 
	\centering 
		\resizebox{1.0\columnwidth}{!}{%
	\begin{tabular}{|c|c|c|c|c|}
		\hline
		Serial & Type of Graph & Memory Requirement  & Algorithm Name & Time Until \\
		No. & & of Each Robot &  &  Dispersion Achieved  \\
		\hline
		\hline
		1. & Rooted Ring & $O(1)$ bits & Rooted-Ring & $O(n)$ rounds \\
		\hline
		2. & Rooted Tree & $O(\log \Delta)$ bits & Rooted-Tree & $O(n)$ rounds  \\
		\hline
		3. & Rooted Graph & $O(\max \lbrace \log \Delta, \log D \rbrace)$ bits & Rooted-Graph-LogDelta-LogD & $O(m)$ rounds \\
		\hline
		4. & Rooted Graph & $O(\Delta)$ bits & Rooted-Graph-Delta & $O(m)$ rounds  \\
		\hline
		5. & Arbitrary Graph & $O(\log \Delta)$ bits & *Arbitrary-Graph & Cover time of graph \\
		\hline
		\multicolumn{5}{|l|}{*Arbitrary-Graph: In this algorithm, robots do not terminate execution of the algorithm, unlike the other algorithms.}\\
		\multicolumn{5}{|l|}{The cover time of a graph lies in the range between $\Omega(n \log n)$ and $O(mn)$.}\\
		\hline
	\end{tabular}
		}
	\label{table:results}
\end{table*}


\vspace{-.3cm}
\subsection{Related Work}


The problem of dispersion of mobile robots on a graph was introduced recently by Augustine and Moses Jr.~\cite{AM18} and studied in different graph classes. In the full version \cite{AM17}, the authors rectified and improved some of their dispersion algorithms. Improvements and rectifications were also independently performed by Kshemkalyani and Ali~\cite{KA18}. Both papers focused on the trade-off between time complexity and memory requirement of robots to solve dispersion deterministically. Our results improve over the previous works \cite{AM17,KA18} in terms of memory requirement with the help of randomness. In particular, our randomized algorithms reduce the memory requirement from $O(\log n)$ bits to $O(\log \Delta)$ bits and the time complexity remains same or is faster (in some cases). While the algorithms in \cite{AM17,AM18} chiefly rely on a timer to signal termination of the algorithm and as such require $\Omega(\log n)$ bits of memory, our algorithms are more event oriented and robots terminate when the termination condition is triggered.   

Dispersion is closely connected to graph exploration by robots; a well-studied problem in the context of mobile robots. In the graph exploration problem, $k$ robots are initially located at a node and the goal is to have the robots collectively visit all nodes in the graph. A number of papers have worked on this problem, however, most of the works are in specific graph classes, such as rings~\cite{DALALLPF15,LABMTB10}, 
 trees~\cite{DFKP02,DKHS06,OS14}, 
  and grids~\cite{BBMR08,DAFPS12}. 
  Several papers consider exploration on general graphs~\cite{BCGX11,BVX14,CFIKP08,DDKPU15,FIPPP05}. However, the model assumptions or the goal of these papers are different from ours and they may produce inefficient solutions to dispersion. For example, the papers  close to our model~\cite{CFIKP08,FIPPP05,R08} only focus on minimizing memory of the robots and as a result the time complexity of their exploration algorithms is very high. Fraigniaud et al.~\cite{FIPPP05} shows that a robot with $\Theta(D\log \Delta)$ bits of memory can explore an anonymous graph, but may take time $O(\Delta^{D+1})$. Cohen et al.~\cite{CFIKP08} considers the model where the nodes also have memory. Then with some initial preprocessing, they solve exploration with less memory bits, but the exploration takes $O(\Delta^{10}m)$ time. Further, the exploration algorithm of Cohen et al. with $O(1)$ bits memory at each node cannot solve dispersion immediately. Diks et al.~\cite{DFKP02} shows that exploration in a tree is possible with $O(\log^2 n)$ bits of memory. Amb{\"u}hl et. al.~\cite{AGPRZ11} improved this memory bound to $O(\log n)$ bits. 
Dynia et al.~\cite{DKHS06} and Ortolf et al.~\cite{OS14} present optimal-time rooted tree exploration algorithms with $k$ robots, but they assume unlimited memory of robots.  Our paper focuses on the trade-off between running time and memory requirement to solve dispersion.     




Another similar problem to dispersion is scattering or uniform deployment of $k$ robots $k \leq n$ on a graph. In the scattering problem on a graph, $k$ robots need to uniformly deploy over $n$ nodes in the graph. Several papers studied the scattering problem on graphs; e.g., on rings \cite{EB11,SMOKM16} and on grids \cite{BFMS11}, but in different settings. 

Most of the above algorithms for graph exploration or scattering on graphs are deterministic. To the best of our knowledge, our paper is the first presenting randomized solutions to dispersion and improve the previous results.


A slightly different way of looking at the dispersion problem is as load balancing in graphs. Load balancing requires nodes to distribute the load over nodes evenly. Here, if we consider robots as the load, then dispersion of robots is similar to load balancing, where the power to move load around the graph lies with the load as opposed to the nodes. Load balancing is a well explored problem, in particular in graphs~\cite{BV86,Cybenko89,MGS98,PV89,SS94}. 
 Our model is closer to diffusion based load balancing~\cite{Cybenko89,MGS98,SS94} with discrete loads~\cite{BV86,PV89}.


\subsection{Organization of Paper}
The rest of the paper is organized as follows. In Section~\ref{sec:prelims}, we introduce the technical preliminaries needed for our results. In Section~\ref{sec:local-leader-election}, we present an important primitive, \emph{Local-Leader-Election}, which we use extensively in our algorithms. In Sections~\ref{sec:rooted-ring},~\ref{sec:rooted-tree}, and~\ref{sec:rooted-graph}, we present our algorithms to achieve dispersion on rooted rings, rooted trees, and rooted graphs respectively. In Section~\ref{sec:arb-graph}, we present a simple memory optimal algorithm to achieve dispersion on arbitrary graphs. 
We show how to extend our algorithms to handle dispersion with an arbitrary number of robots in Section~\ref{sec:arb-k}. The lower bound on memory per robot is presented in Section~\ref{sec:lower-bounds}.  Finally, in Section~\ref{sec:conc}, we present conclusions and some future work.


\section{Technical Preliminaries}
\label{sec:prelims}

We consider a connected undirected graph of $n$ nodes, $m$ edges, diameter $D$, and maximum degree of any node $\Delta$. The nodes are anonymous, i.e. they do not have unique labels. For every edge connected to a node, the node has a corresponding port number for the edge. The same edge may have different port number assigned to it at each of its attached nodes. For every node, there exists a total ordering on the port numbers from that node. 
A robot with $x$ bits of memory has access to $2^x$ ports of that node. For a given node with $y$ ports, if $2^x \geq y$, then the robot has access to all ports of the node. When $2^x < y$, the robot only has access to a subset of the ports, where the exact subset of ports is chosen arbitrarily by nature.\footnote{The robot's memory restricts it to only use a subset of the available ports when determining which port to move through. Importantly, the robot does not know that there are more ports than it is seeing. Thus it cannot purposely choose which subset of ports to see. We call this lack of control  ``by nature".} 
Thus, for a given node with $\Delta$ ports, any robot needs at least $\log \Delta$ bits of memory in order to access all ports.\footnote{This does not necessarily give a $\Omega(\log \Delta)$ memory lower bound for dispersion. We discuss this further in the lower bound section (Section~\ref{sec:lower-bounds}).}

We assume a synchronous system, i.e. time progresses in rounds, and each robot knows when the current round ends and a new round starts, although robots may not know the round number. Each round proceeds as follows: (i) First, robots colocated at the node exchange messages with each other and perform local computation. (ii) Second, robots move through a port of current node and reach a new node. Robots may also choose to stay at the current node. Note that in step (i) of the round, we consider local computation and message exchange to be bounded, but free\footnote{This can be considered a realistic assumption when the time taken for a robot to move through an edge is significantly more than the time taken for either local message exchange or local computation.}.

Robots are anonymous, i.e. they do not have unique labels. Each robot has a limited amount of memory for computation and to store information. The exact limit depends on the algorithm to be run and is explicitly given in each section of the paper. Each robot has access to a fair coin that be used to generate an infinite number of random bits. However, the number of random bits that can be stored and used for any purpose is limited by the robots memory. When a robot is present at a node, it can access the port numbers  of that node, subject to memory restrictions as defined earlier. The robot can only view other robots colocated at the same node as it and cannot see anything beyond its current node (its ``view" of the graph is limited). Robots do not know the value of $n$, $m$, $D$, or $\Delta$. Note that our algorithms require do not require robots to know the actual values of $D$ and $\Delta$, but require the robots to have enough memory store either $O(\log D)$, $O(\log \Delta)$, or $O(\Delta)$ bits according to the algorithm in question. Thus robots may have an upper bound on those values but do not explicitly know those values. 

We characterize the efficiency of solutions to dispersion along two metrics. First, how many bits of memory is each robot required to have. Second, what is the running time of the algorithm until dispersion is achieved. For all algorithms, save the algorithm in Section~\ref{sec:arb-graph}, robots execute the dispersion algorithm and then terminate within the running time we specify for the algorithm. For the algorithm in Section~\ref{sec:arb-graph}, we allow robots to be active indefinitely.

We now present several definitions of terms we use in the paper. We call a graph a \textit{rooted graph} if all robots are initially placed at one node of that graph called its root. A similar definition applies for specific types of graphs such as rings and trees. We say that a robot \textit{settles at a given node} if that robot chooses to stay at that node in the final dispersion configuration. We call a node with a robot that settles on it a \textit{settled node}. The algorithm in Section~\ref{sec:arb-graph} is based on random walks. A {\em simple random walk} in an undirected graph is defined as: in each step, the walk chooses a random adjacent edge from the current node and moves to that neighbor. The probability of choosing a random neighbor $u$ from the current node $v$ is $1/d(v)$, where $d(v)$ is the degree of $v$. The {\em cover time} of a random walk is defined to be the time required by the random walk to visit all the nodes in the graph. It is known that the cover time of any graph is bounded by $O(mn)$ \cite{AKLLR79}. We refer the reader to the survey \cite{LL93} for details on random walks and cover time. 

We now formally define dispersion of $k$ robots on an $n$ node graph. Initially, $k$ robots are arbitrarily placed on the graph. Dispersion asks that robots move around the graph to reach a configuration such that at most $\lceil k/n \rceil$ robots are present at any given node.

\section{Local Leader Election}
\label{sec:local-leader-election}

In this section, we describe a procedure which we use throughout the rest of the paper. The procedure allows any number of $k$ robots co-located at an unsettled node to choose exactly one leader (robot) for the node within one round of an algorithm. Importantly, each robot only requires $O(1)$ bits of memory and access to a random number generator in order to execute the algorithm. We first describe the algorithm and prove our claims on it. We subsequently discuss the applications of the algorithm that immediately arise. In order to differentiate between different instances of communication occurring between robots at a node within the same round, we refer to each instance of communication among at most $k$ robots as one sub-round and use that terminology while describing our algorithm. Note that, as per our model assumptions, any amount of communication is allowed to take place between robots within a single round of the system, so long as the amount of communication is bounded. As we shall see, our procedure satisfies that requirement with high probability.

\subsection{Algorithm Local-Leader-Election} Each robot starts off as a candidate for leader. In every sub-round, a robot that is a candidate leader flips a fair coin. If heads, it broadcasts that it is alive to other robots. If tails (and at least one other robot broadcasts in that sub-round) it stops being a candidate for leader and doesn't broadcast anymore. If tails and no other robot broadcasts in that sub-round, it remains a candidate for leader. This process is repeated until exactly one robot broadcasts in a given sub-round. Then all robots know that that robot is the leader, and it is chosen as the robot which settles down at the given node. Subsequently, all other robots can then move to other nodes according to a given algorithm. Note that if it occurs that no robot broadcasts in a given sub-round, that sub-round is ignored and all robots that were still alive previously broadcast again.

\begin{theorem}
\emph{Local-Leader-Election} can be run by multiple robots, each having $O(1)$ bits of memory and co-located at a node, to select a common leader within one round of the system.
\end{theorem}

\begin{proof}
We first show that the algorithm can be run by robots and completes within one round of the system. Then we argue about its memory complexity.

It is easy to that \emph{Local-Leader-Election} takes $O(\log n)$ sub-rounds on expectation for a leader to be chosen. Applying a Chernoff bound, it is also easy to see that it takes $O(\log n)$ sub-rounds with high probability for termination. Thus the number of sub-rounds is bounded with high probability. Recall that any amount of bounded communication between co-located robots is allowed in one round of the system. Thus the algorithm successfully executes within one round of the system.

Each robot requires only $O(1)$ bits because a robot needs $1$ bit to check if it's a candidate leader, $1$ bit to check if it's the leader, and $2$ bits to check if it heard $0, 1,$ or more than $1$ robot broadcast in a given sub-round. \qed
\end{proof}

\subsection{Applications}
\emph{Local-Leader-Election} can be directly applied to past deterministic algorithms in \cite{AM18} to replace the use of $O(\log n)$ bits to compare multiple robots to decide which one settles at a node. In addition, if the termination condition is relaxed, some of the resulting algorithms use dramatically less memory as a result. We list the improvements to algorithms in \cite{AM18}, as a result of these two modifications, below:
\begin{enumerate}
\item Algorithm \emph{Path-Ring-Tree-LogN} achieves dispersion of $n$ robots in $O(n)$ rounds on paths, trees, and rings when robots have $O(\log \Delta)$ bits of memory and do not terminate.
\item Algorithm \emph{Rooted-Graph-LogN} achieves dispersion of $n$ robots in $O(m)$ rounds on rooted graphs when robots have $O(\log \Delta)$ bits of memory and do not terminate.
\end{enumerate}

\section{Dispersion on Rooted Rings}
\label{sec:rooted-ring}

In this section, we describe our algorithm to achieve dispersion of $n$ robots on a rooted ring in $O(n)$ rounds when robots have $O(1)$ bits of memory. This does not contradict our lower bound because here $\Delta$ is a constant so $O(\log \Delta) = O(1)$. Recall that for a ring, any algorithm to achieve dispersion takes at least $\Omega(n)$ rounds because the diameter of the graph is $n/2$. Thus our algorithm is asymptotically time optimal.

\subsection{Algorithm Rooted-Ring} Each robot performs a traversal of the ring in a deterministic manner until it becomes the leader of the node it is at, as chosen in \emph{Local-Leader-Election}. Once it becomes the leader of that node, it settles down and terminates execution of the algorithm. 

The traversal of the ring is done in the following manner. Initially have robots move through port $0$. Subsequently, if the robot enters a node through port $i$ and it does not become the leader, have it leave through port $i+1 \mod 2$.

\begin{theorem}
Algorithm \emph{Rooted-Ring} can be run by $n$ robots with $O(1)$ bits of memory each to ensure dispersion occurs in $O(n)$ rounds on rooted rings.
\end{theorem}

\begin{proof}
It is easy to see that the entire ring is traversed by robots in $O(n)$ rounds. Further, each robot terminates as soon as it becomes a leader. Finally, each node has exactly one leader robot assigned to it. Thus dispersion is achieved in $O(n)$ rounds. The only memory requirement is a bit to remember which port the robot entered the node through and $O(1)$ bits to execute \emph{Local-Leader-Election}. Thus each robot only requires $O(1)$ bits of memory. \qed
\end{proof}

\section{Dispersion on Rooted Trees}
\label{sec:rooted-tree}

In this section, we describe an algorithm to achieve dispersion of $n$ robots on a rooted tree in $O(n)$ rounds when each robot has $O(\log \Delta)$ bits of memory.
\vspace{-0.3cm}
\subsection{Algorithm Rooted-Tree}\label{subsec:rooted-tree-alg}
The algorithm has two phases of execution. In the first phase, the algorithm has every robot perform a deterministic depth first search (DFS) in order to uniquely settle down at a node. In the second phase, the final robot to settle down then backtracks to the root of the tree and performs a second DFS to inform each robot to terminate execution.

In the first phase, each robot that does not settle down performs a DFS in the following manner. It remembers the port $i$ that it entered the node through. It then leaves through port $(i+1) \mod \delta$, where $\delta$ is the local degree of the node. Initially at the root, let the robots move through port $0$ (since the robots did not initially enter the root through any node). At each empty node, a robot is chosen by \emph{Local-Leader-Election} to settle down at it. This node remembers the port it entered the node through and we call the port the \textit{parent pointer}.

In the second phase, the final robot to settle, $x$, changes its status to reflect that phase two has begun and backtracks to the root of the tree and then performs a second DFS until it finally settles down again and terminates. The robots it comes in contact with terminate when a special condition is met, as defined below. Consider the set of nodes on the path from the node where the last robot settled to the root and call it $R$. Let $S$ represent the set of all nodes in the graph. For every node $u\in S \setminus R$, the robot at $u$ terminates execution when $x$ leaves $u$ through $u$'s parent pointer. For each node $u\in R$, we have the robot at $u$ remember the port that $x$ backtracked through to reach the root, i.e. the port $x$ entered node $u$ through as it backtracked, and we call that port the \textit{pointer to final node}.\footnote{It is possible for the robot at $u$ to differentiate $x$ from just another robot executing phase one of the DFS because $x$ has changed its status to reflect that the second phase has begun.} Once $x$ passes through the pointer to final node of $u$, $u$ terminates execution. Finally, when $x$ reaches its empty node, it settles down and terminates execution.

\begin{theorem}\label{the:rooted-tree}
Algorithm \emph{Rooted-Tree} can be run by $n$ robots with $O(\log \Delta)$ bits of memory each to ensure dispersion occurs in $O(n)$ rounds on rooted trees.
\end{theorem}

\begin{proof}
We first prove that the execution of \emph{Rooted-Tree} results in dispersion being achieved and all robots terminating within $O(n)$ rounds of the start of the algorithm. Subsequently, we argue that each robot requires only $O(\log \Delta)$ bits of memory. 

It is clear that the first phase results in robots performing a DFS until an empty node is found to settle down in. Thus at the end of phase one, there is exactly one robot on each node. Note that we use notation from the algorithm. We now show that in phase two, due to the careful way we trigger termination of the algorithm in robots, all robots will terminate at the end of the DFS of $x$ and $x$ will end up back at the node it originally settled at. We first show this for robots at the nodes in the set $S \setminus R$ and then for the remaining robots.

\begin{lemma}\label{lem:subtree-execution}
For every node $u \in S \setminus R$, the robot in every node in the subtree rooted at $u$ terminates execution before the robot in $u$ terminates execution.
\end{lemma}

\begin{proof}
Consider only the subtree rooted at $u$ and let $S' \subseteq S \setminus R$ denote the set of nodes of the subtree including $u$. Let the maximum depth of any node in this subtree be $d$. $u$ is at depth $0$ in this subtree. Now, we prove by induction on the depths of nodes in the subtree that for any node $v \in S'$ at depth $d'$, the following hypothesis holds. The robots in all descendant nodes of $v$ have terminated before the robot in $v$ terminates.

For a node at depth $d$, i.e. a leaf node, the hypothesis holds trivially because it has no descendants. 
Let the hypothesis hold true for all nodes at some depth $d'$ in the subtree. We now prove that it holds for all nodes at depth $d'-1$. Consider a node $w \in S'$ at depth $d'-1$. Before moving through the parent pointer of $w$, $x$ will have explored each child of $w$ and moved through that child's robot's corresponding parent pointer. Thus the robot at each child is triggered to end termination of the algorithm before the robot at $w$ is triggered to end execution. Thus the invariant holds true for $w$. \qed
\end{proof} 

For every node $u \in R$ excluding the final node that $x$ settles down at, it is clear to see that once $x$ passes through $u$'s pointer to the final node, $u$ will not be visited again. So the robot at $u$ can terminate without a problem. And finally, once $x$ reaches its empty node, it will also terminate execution. Thus, we see that for all nodes in $S$, after $x$ completes its DFS in phase two, all robots at those nodes will terminate execution. The time taken to perform two DFS's on a tree and have $x$ move to the root from a settled node is $O(n)$ rounds. Thus the algorithm successfully completes in $O(n)$ rounds.

Regarding the memory requirements of each robot, $O(\log \Delta)$ bits are required for a parent pointer and pointer to the final node. $O(1)$ bits are required to remember which phase the robot is in, to denote whether the robot is in the root of the tree, to denote whether a robot becomes the exploring robot $x$, and to perform \emph{Local-Leader-Election}. So totally, each robot requires $O(\log \Delta)$ bits of memory. \qed
\end{proof}

\vspace{-0.5cm}
\section{Dispersion on Rooted Graphs}
\label{sec:rooted-graph}
 
In this section we describe two algorithms to achieve dispersion of $n$ robots on a rooted graph in $O(m)$ rounds. One algorithm requires robots to have $O(\max\lbrace \log \Delta, \log  D\rbrace )$ bits of memory each while the other has a requirement of $O(\Delta)$ bits of memory.
\vspace{-0.3cm}
\subsection{Algorithm Rooted-Graph-LogDelta-LogD}
This algorithm can be thought of as an extension to the algorithm \emph{Rooted-Tree} found in Section~\ref{subsec:rooted-tree-alg}. Similar to that algorithm, \emph{Rooted-Graph-LogDelta-LogD} proceeds in two phases. In the first phase, robots again perform a deterministic DFS in order to find nodes to settle down at. However, the key difference between this algorithm and \emph{Rooted-Tree} lies in how this algorithm deals with cycles in the graph. In the second phase, again the last settled robot goes to the root and performs a second DFS, triggering other robots to terminate execution of the algorithm in the process. 

In the first phase, robots again perform a DFS by remembering the port $i$ they entered the node through and subsequently leaving through port $(i+1) \mod \delta$, where $\delta$ is the local degree of the port. For the root node, robots initially move through port $0$. At each empty node along the way, robots perform \emph{Local-Leader-Election} to choose one robot to settle down at that node. If a robot is exploring and comes across a node with a robot already settled on it, the exploring robot backtracks to its previous node and tries the next port from that node. Finally, the last robot settles down at the last node, marking the end of phase one. Note that settled robots maintain a parent pointer which records the port through which they first entered the given node. For the root, the value of its parent pointer is null. Also, each robot (both settled and exploring) maintain a counter indicating its distance from the root with respect to the tree of nodes formed by the DFS.

In the second phase, the last robot to settle down, $x$, changes its state info to indicate that it is in phase two and makes its way to the root of the graph. From here it performs a DFS, similar to that done in phase one. However, in this DFS, cycles are detected when the robot moves from a node at distance $\ell$ from the root to a node at distance $< \ell$ from the root. In such a case, the robot backtracks to its previous node and proceeds through the next available port. Similar to \emph{Rooted-Tree}, consider all nodes along the path from the root to the node that $x$ finally settles at. Call this set $R$. Let every robot belonging to a node in $R$ except for $x$ maintain a pointer to the final node, indicating the port through which $x$ must go to return to the node it must settle at. Once $x$ moves through this port in the course of the DFS, all robots belonging to nodes in $R$ save $x$ itself terminate execution. Let the set of all nodes in the graph be $S$. For a robot belonging to a node $u \in S \setminus R$, once $x$ passes from $u$ to $u$'s parent via $u$'s parent pointer, the robot at $u$ terminates execution of the algorithm. Finally, $x$ completes the DFS, returns to the node it must settle at, and terminates execution.

\begin{theorem}\label{the:rooted-graph-logdelta-logd}
Algorithm \emph{Rooted-Graph-LogDelta-LogD} can be run by $n$ robots with $O(\max\lbrace \log \Delta, \log  D\rbrace )$ bits of memory each to ensure dispersion occurs in $O(m)$ rounds on rooted graphs.
\end{theorem}

\begin{proof}
This proof is very similar to the proof of Theorem~\ref{the:rooted-tree}. We first prove that all robots running \emph{Rooted-Graph-LogDelta-LogD} achieve dispersion in $O(m)$ rounds. Then we argue that each robot requires $O(\max\lbrace \log \Delta, \log  D\rbrace )$ bits of memory to execute the algorithm.

To prove our claim on dispersion, we make use of the following Claim. We omit the proof as the Claim and its proof are very similar to Lemma~\ref{lem:subtree-execution} and its proof.

\begin{claim}\label{claim:subtree-execution-rootedgraph1}
For every node $u \in S \setminus R$, the robot in every node in the subtree rooted at $u$ terminates execution before the robot at $u$ terminates execution.
\end{claim}

Note that the proof of Claim~\ref{claim:subtree-execution-rootedgraph1} requires one extra argument in addition to the argument required in the proof of Lemma~\ref{lem:subtree-execution}. We must show that $x$ will immediately backtrack when a cycle is detected and not trigger the termination condition of a robot out of order of the DFS by accidentally further exploring through that robot's parent pointer. Our cycle checking strategy requires robots to maintain their distance from the root. $x$ can easily identify that it is in a cycle if it moves from a node at distance $\ell$ to one at distance $< \ell$. Importantly, in a DFS, cross edges do not exist but only forward and back edges. This means that $x$ moved to an ancestor of the node. $x$ would not have yet moved through the parent pointer of the robot attached to the ancestor, and thus would not have triggered that robot to terminate execution. So, $x$ can communicate with the robot at the ancestor, discover the distance of the ancestor from the root, discover that it is in a cycle, and backtrack immediately.

Thus, at the end of phase two, $x$ will successfully complete execution of the DFS. At the end of the DFS, all robots would be triggered to terminate execution. Furthermore, we know that a DFS on a graph takes $O(m)$ rounds, so that is the execution time of the algorithm.

As to the memory requirement of robots. Each robot must use $O(\log \Delta)$ bits to store information about parent pointer and pointer to final node. Additionally, $O(\log D)$ bits are required to store information about the distance to root. Finally, $O(1)$ bits are needed to store information about which phase a robot is in, whether the robot is in the root or not, whether the robot is the exploring robot $x$ or not, and to execute \emph{Local-Leader-Election}. Therefore, each robot requires $O(\max\lbrace \log \Delta, \log  D\rbrace )$ bits of memory to execute the algorithm. \qed
\end{proof}

\subsection{Algorithm Rooted-Graph-Delta}
This algorithm is a variation of \emph{Rooted-Graph-LogDelta-LogD} that provides a memory trade-off: instead of needing $O(\log D)$ bits of memory, this algorithm requires $O(\Delta)$ bits of memory. This algorithm again runs in two phases, with the goals of each phase being the same as that of the previous algorithm. We focus only on the variation between the two algorithms now.

In the previous algorithm, each settled robot was required to remember its distance from the root. In this algorithm, instead we require each settled robot to remember which of its ports lead to forward edges and which do not. This is done by maintaining a bit string of size at most $\Delta$ bits where each bit from LSB to MSB corresponds to one of the ports leading out of that node. Let us call it \textit{list of forward ports}. Initially all bits except parent pointer's bit are set to one to indicate that all those ports possibly lead to forward edges. In the course of the DFS in phase one, if a robot uses a port from a node $u$ and then realizes it is in a cycle, the robot will backtrack to $u$, inform the robot at $u$ about the given port, and then move on with the DFS. The robot at $u$ sets the corresponding bit in the list of forward ports to zero. Thus at the end of phase one, all settled robots have an accurate list of forward ports. Now, in phase two, when $x$ performs its DFS, at a given node it only considers those ports whose corresponding bit in the list of forward ports is one.

\begin{theorem}
Algorithm \emph{Rooted-Graph-Delta} can be run by $n$ robots with $O(\Delta)$ bits of memory each to ensure dispersion occurs in $O(m)$ rounds on rooted graphs.
\end{theorem}

\begin{proof}[Proof Sketch]
The proof of this theorem is identical to that of Theorem~\ref{the:rooted-graph-logdelta-logd}, so we omit details. However, we focus on two things: the proof that $x$ does not end up in a cycle during the phase two DFS and the memory requirements of robots.

In phase two of the algorithm, $x$ does not need to traverse an edge to discover if it is a back edge. Instead, at a given node $u$, it needs only rely on the list of forward ports maintained by the robot at $u$, which will be properly built in phase one of the algorithm. Therefore, the second DFS will be successful and all robots will terminate execution of the algorithm in $O(m)$ rounds.

Instead of using $O(\log D)$ bits of memory to remember the distance from root, each robot is required to maintain the list of forward ports, which is a bit string of size at most $\Delta$. Thus each robot requires $O(\Delta)$ bits of memory. \qed
\end{proof}

\vspace{-0.3cm}
\section{Dispersion on Arbitrary Graphs (without Termination)}
\label{sec:arb-graph}
 
\vspace{-0.3cm}
In this section, we assume that the robots are initially arbitrarily located at nodes in the graph (i.e., not necessarily at a single node). We describe a simple randomized algorithm that can be run by robots to achieve dispersion and requires each robot to have $O(\log \Delta)$ bits of memory each. The algorithm is a Las Vegas style randomized algorithm in that the time until dispersion is achieved is variable and is bounded by the cover time of the graph with high probability . 

\vspace{-0.3cm}
\subsection{Algorithm Arbitrary-Graph}
The idea of the algorithm is that each robot performs a simple random walk on the graph in parallel until it finds an empty node that it can settle at. The algorithm is described below in more detail.
 Each robot performs a random walk on the graph in parallel. During the random walk, if a robot finds an empty node, it settles down in that node. If multiple robots are present at an empty node in the same round, they compute a local leader using the procedure \emph{Local-Leader-Election} and the leader settles at that node. Each robot performs the random walk until it settles down. A settled robot stays active indefinitely since it must inform other exploring robots about the occupancy of the node.

\begin{theorem}\label{thm:main-arb-graph}
Suppose $n$ robots are placed arbitrarily over an $n$ node graph. Then Algorithm \emph{Arbitrary-Graph} solves dispersion with $O(\log \Delta)$ bits of memory in cover time of the graph with high probability. The robots are active indefinitely.   
\end{theorem}

\begin{proof}
First of all, a robot can perform a simple random walk with $O(\log \Delta)$ bits of memory as there are at most $\Delta$ adjacent edges at any node. Hence the robot can pick a random adjacent edge with $O(\log \Delta)$ bits of memory (i.e., it can generate a random number from $1$ to $\Delta$ with $O(\log \Delta)$ bits of memory). 

Consider a particular robot exploring the graph by performing a random walk. Since the random walk is independent of all the other robots' random walks, the robot visits all the nodes in the graph by the cover time of the graph with high probability (see the definition of the cover time in Section~\ref{sec:prelims}). Hence, with high probability, by at most the cover time of the graph, the robot settles at some node. Since every robot performs a random walk in parallel (until it settles down) and independently, every robot will settle down after the cover time of the graph with high probability. Hence dispersion is achieved in cover time of the graph with high probability. The cover time of a graph lies in the range between $\Omega(n \log n)$ and $O(mn)$ depending on the graph structure. 

Whenever a robot settles at some node, the robot has to stay active until all the robots settle down. This is because the settled robot needs to inform other exploring robots that the node is already occupied; otherwise multiple robots may be settled down at a single node and dispersion will not be achieved. Since it takes cover time of the graph until all robots settle down, and the cover time of any graph is $\Omega(n\log n)$ \cite{DA89,BK89}, a robot cannot maintain a counter to count the rounds until cover time is achieved with only $O(\log \Delta)$ bits of memory and hence cannot terminate after the cover time. Thus, all the settled robots need to stay active for an indefinite number of rounds in order to achieve dispersion.    \qed  
\end{proof}

Note that the random walk based exploration algorithm outperforms the $O(m)$ time algorithms in several graph classes. For example, consider regular expander graphs. The cover time of a regular expander graph is $\Theta (n\log n)$ \cite{BK89}. However, in dense regular expander graphs, the number of edges is $\omega (n\log n)$ and it could be as high as $O(n^2)$ (e.g., a complete graph). In fact, the  \emph{Arbitrary-Graph} algorithm is asymptotically faster than deterministic algorithms with more memory from \cite{AM18} in the graphs where $m=\omega(\text{cover-time})$. However, the algorithm is non-terminating. At the same time, \emph{Arbitrary-Graph} requires robots to only have  $O(\log \Delta)$ bits of memory. Moreover, since the random walks (corresponding to the robots) are independent of each other, the algorithm also works in an asynchronous system\footnote{It is required that \emph{Local-Leader-Election} works without issue in that setting. This is the case.}.   

\section{Extending Algorithms to Arbitrary $k$}
\label{sec:arb-k}

In this section, we describe how to extend our previous algorithms to work with an arbitrary number of robots. That is, we want to achieve dispersion of $k$ robots on $n$ nodes, for any positive integer value of $k$. There are two difficulties inherent in extending results to an arbitrary $k$, depending on whether $k < n$ or $k>n$. When $k<n$, if an algorithm relies on a certain condition to be met before robots terminate, we must ensure that this condition is still met even with less than $n$ robots participating. When $k>n$, we must figure out how to have robots settle in stages, because we do not want to  maintain a counter to allow $\lceil k /n \rceil$ robots to settle at each node because the memory requirement will be $O(\log k)$, which could be arbitrarily large.

\emph{Rooted-Ring} works for any arbitrary $k$. At the end of a given round, one robot has settled and terminated execution of the algorithm and the remaining robots are colocated at the next node. It takes one round to settle one robot, and the ring is settled node by node. Thus, for $k$ robots, we can achieve dispersion with \emph{Rooted-Ring} in $O(k)$ rounds where each robot requires $O(\log \Delta)$ bits of memory.

\emph{Rooted-Tree} works without any modifications when $k<n$. This is because the last robot to settle down initiates phase two of the algorithm. It is easy for a robot to detect this because no other robot will participate in \emph{Local-Leader-Election} with it. Since only $k$ nodes of the graph are explored in phase one or two, \emph{Rooted-Tree} only takes $O(k)$ rounds to complete. When $k>n$, we have robots perform a second check to see if it is the last robot to settle. In a complete DFS traversal of a tree, a robot will reach the last node of the traversal when it has traversed the last port of each node from the root to a leaf node\footnote{Since ports are ordered, each robot can determine which port is last in the order. Furthermore, this ordering is unique to each node, so different robots will see the same ordering at each node.}. Thus, if a robot performing the DFS maintains a flag that indicates whether the robot traversed this path, we can detect the final node in the DFS traversal. Specifically, the robot sets the flag to true when at the root, and changes it to false if it traverses a port other than the last port of a given node. If the robot reaches a leaf node and the flag is true, and the robot is selected by \emph{Local-Leader-Election} to settle down at the leaf node, then it knows that it is settling down at the last empty node in the tree. All other robots that got to this node but did not settle down go to the root node and wait. The last robot to settle down executes phase two of the algorithm as usual. Meanwhile, the robots waiting at the root node will execute a new iteration of the algorithm, i.e. start phase one, once the exploring robot in phase two passes through the final port of the root. This delay in start time  guarantees that the exploring robot will trigger other robots in the tree to terminate execution in time for the robots executing phase one to re-populate the tree with settled robots. The above process is repeated $\lfloor k/n \rfloor$ times and takes $O(n)$ rounds for each repetition and with an additional repetition done by some $k - n*\lfloor k/n \rfloor$ robots that takes $O(k - n*\lfloor k/n \rfloor)$ rounds, for a total of $O(k)$ rounds to achieve dispersion of all robots. Thus, for any $k$, \emph{Rooted-Tree} achieves dispersion in $O(k)$ rounds and requires robots to have $O(\log \Delta)$ bits of memory each.

\emph{Rooted-Graph-LogDelta-LogD} and \emph{Rooted-Graph-Delta} both work when $k<n$ without any modifications. When $k>n$, both work using the extension described in the previous paragraph on \emph{Rooted-Tree}. However, since both algorithms take $O(m)$ rounds to settle $n$ robots, the total running times are both $O(mk/n)$. However, an interesting change to the \emph{Rooted-Graph-Delta} algorithm can ensure a running time of $O(m + k)$. The change is to add an extra bit to each robot that indicates if they are participating in an ``even" or ``odd" set of phases of the algorithm. Order the robots in sets of $n$ robots referencing the set of phases in which they are settled and terminate execution of the algorithm. If the set is an odd (even) number set in this order, it is called odd (even) set. Now, for every set of phases $a$ and $b$ where $a$ immediately precedes $b$ in the order, order of phases is as follows. First, phase one of set $a$ occurs, then phase one of set $b$, then phase two of set $a$, then phase two of set $b$.\footnote{For $3$ sets of phases $a$, $b$, and $c$, sequence is: phase one of $a$, phase one of $b$, phase two of $a$, phase one of $c$, phase two of $b$, phase two of $c$. This sequence can be easily seen now for more sets.} A node can be identified as a possible soon to be empty node (in the sense that the settled robot terminates execution) if only a robot with a different phase (odd vs. even) robot is settled at it. Thus in phase one of set $b$, have each settled robot copy the values of the at most $\Delta$ bits of the previous set's robot at that node indicating which ports lead to back edges. Thus, only in phase one of the first set of nodes will the DFS take $O(m)$ rounds, and subsequent DFS's take only $O(n)$ rounds each for the remaining $O(k/n)$ sets of nodes (the algorithm is slightly tweaked so that future sets of robots take advantage of this info). Thus the running time becomes $O(m + k)$ rounds. 

\emph{Arbitrary-Graph} works without any modifications when $k<n$. When $k>n$, in general the algorithm will not work because robots never terminate execution of the algorithm. Thus, unlike the previous algorithms, we cannot have robots work in stages where the first set of robots settles down and terminates, then the next set, and so on. Instead, we would need a counter of $O(\lceil k/n \rceil)$ bits to count how many robots have already settled at a node and settle down (if chosen by \emph{Local-Leader-Election}) if that counter is $< \lceil k/n \rceil$. This leads to the following special case where dispersion is possible. When $n < k \leq \Delta^c n$, where $c$ is a positive constant, we can achieve dispersion using \emph{Arbitrary-Graph}, modified with a counter as earlier described, in the cover time of the graph with high probability where each robot needs $O(\log \Delta)$ bits of memory.

\section{Lower Bound on Memory}
\label{sec:lower-bounds}

In \cite{AM18}, they showed that, assuming all robots have the same memory, a lower bound of $\Omega(\log n)$ bits is required for dispersion when considering deterministic algorithms. The bound resulted from an argument that robots needed enough bits to uniquely choose a robot from a set of robots at each node. However, as we see later in this paper, we are able to circumvent this with the use of randomness. Now we argue another lower bound in the presence of randomness. 

\begin{theorem}
Consider $k$ robots trying to achieve dispersion on an $n$ node graph. Assuming all robots have the same amount of memory, robots require $\Omega(\log \Delta)$ bits of memory each for any randomized algorithm to achieve dispersion on any graph.
\end{theorem}

\begin{proof}
We first describe a situation where robots containing $o(\log \Delta)$ bits of memory will be unable to achieve dispersion. We then show that for all algorithms that attempt to achieve dispersion using $o(\log \Delta)$ bits of memory, we can arrive at this situation.

Consider any number of robots present at a given node with degree $O(\Delta)$. If each robot has $o(\log \Delta)$ bits of memory, it is impossible for any of them to individually access the entire list of possible ports to move through. Since the selection of ports by each of these robots is decided by nature, it may then occur that one particular port is never chosen in any of the subsets of ports. Let us focus on such a port.  

Let this port lead to an edge which acts as a cut between the set of nodes with robots currently on them and the set of empty nodes.  If $k < n$, additionally assume that the set of nodes with robots on them is of size $\leq k-1$. Thus, the robots being unable to traverse that port prevents dispersion from being achieved. For any given algorithm, we can construct a graph such that there exists a node with associated cut edge satisfying the above description and all the robots start on nodes on the side of the cut with the node. Thus, for any algorithm, dispersion is impossible when robots have $o(\log \Delta)$ bits of memory each. \qed
\end{proof}

\section{Conclusion and Future Work}
\label{sec:conc}

In this paper, we showed how to achieve dispersion on various types of graphs using less memory than required by other algorithms in the literature so far. Importantly, we showed how to leverage randomness in a novel way in the form of the \emph{Local-Leader-Election} algorithm and utilize this primitive to reduce memory requirements. There are several interesting lines of research that result from this paper. We present two open problems of interest below.

\textbf{Open Problem 1:} All algorithms in our paper, save \emph{Arbitrary-Graph}, work only for rooted versions of different types of graphs. The trade-off when implementing \emph{Arbitrary-Graph} is that robots then must stay active indefinitely. Is it possible to develop algorithms for non-rooted versions of the graphs in question without requiring robots to stay active indefinitely?

\textbf{Open Problem 2:} Our algorithms for rooted graphs require robots to have possibly $\omega(\log \Delta)$ bits of memory each, depending on the values of $\Delta$ and $D$. Is it possible to develop algorithms with tighter upper bounds for rooted graphs?



\bibliographystyle{splncs04}
\bibliography{references} 

\begin{thebibliography}{10}
\providecommand{\url}[1]{\texttt{#1}}
\providecommand{\urlprefix}{URL }
\providecommand{\doi}[1]{https://doi.org/#1}

\bibitem{DA89}
Aldous, D.J.: Lower bounds for covering times for reversible markov chains and
  random walks on graphs. Journal of Theoretical Probability  \textbf{2}(1),
  91--100 (1989)

\bibitem{AKLLR79}
Aleliunas, R., Karp, R.M., Lipton, R.J., Lov{\'{a}}sz, L., Rackoff, C.: Random
  walks, universal traversal sequences, and the complexity of maze problems.
  In: Proc. of the 20th Annual Symposium on Foundations of Computer Science
  (FOCS). pp. 218--223 (1979)

\bibitem{AGPRZ11}
Amb{\"u}hl, C., G{\k{a}}sieniec, L., Pelc, A., Radzik, T., Zhang, X.: Tree
  exploration with logarithmic memory. ACM Transactions on Algorithms (TALG)
  \textbf{7}(2), ~17 (2011)

\bibitem{AM18}
Augustine, J., Moses~Jr., W.K.: Dispersion of mobile robots: {A} study of
  memory-time trade-offs. In: Proceedings of the 19th International Conference
  on Distributed Computing and Networking, {ICDCN} 2018, Varanasi, India,
  January 4-7, 2018. pp. 1:1--1:10. ACM (2018)

\bibitem{AM17}
Augustine, J., Moses~Jr., W.K.: Dispersion of mobile robots: A study of
  memory-time trade-offs. CoRR  \textbf{abs/1707.05629} ((v4) 2018),
  \url{http://arxiv.org/abs/1707.05629}

\bibitem{BBMR08}
Baldoni, R., Bonnet, F., Milani, A., Raynal, M.: On the solvability of
  anonymous partial grids exploration by mobile robots. In: Proc. of
  International Conference On Principles Of Distributed Systems (OPODIS). pp.
  428--445 (2008)

\bibitem{BFMS11}
Barriere, L., Flocchini, P., Mesa-Barrameda, E., Santoro, N.: Uniform
  scattering of autonomous mobile robots in a grid. International Journal of
  Foundations of Computer Science  \textbf{22}(03),  679--697 (2011)

\bibitem{BV86}
Bodlaender, H.L., van Leeuwen, J.: Distribution of Records on a Ring of
  Processors. Department of Computer Science, University of Utrecht (1986)

\bibitem{BCGX11}
Brass, P., Cabrera-Mora, F., Gasparri, A., Xiao, J.: Multirobot tree and graph
  exploration. IEEE Transactions on Robotics  \textbf{27}(4),  707--717 (2011)

\bibitem{BVX14}
Brass, P., Vigan, I., Xu, N.: Improved analysis of a multirobot graph
  exploration strategy. In: Control Automation Robotics \& Vision (ICARCV),
  2014 13th International Conference on. pp. 1906--1910. IEEE (2014)

\bibitem{BK89}
Broder, A.Z., Karlin, A.R.: Bounds on the cover time. Journal of Theoretical
  Probability  \textbf{2}(1),  101--120 (1989)

\bibitem{CFIKP08}
Cohen, R., Fraigniaud, P., Ilcinkas, D., Korman, A., Peleg, D.: Label-guided
  graph exploration by a finite automaton. ACM Transactions on Algorithms
  (TALG)  \textbf{4}(4), ~42 (2008)

\bibitem{Cybenko89}
Cybenko, G.: Dynamic load balancing for distributed memory multiprocessors.
  Journal of Parallel and Distributed Computing  \textbf{7}(2),  279--301
  (1989)

\bibitem{DALALLPF15}
Datta, A.K., Lamani, A., Larmore, L.L., Petit, F.: Enabling ring exploration
  with myopic oblivious robots. In: Parallel and Distributed Processing
  Symposium Workshop (IPDPSW). pp. 490--499. IEEE (2015)

\bibitem{DDKPU15}
Dereniowski, D., Disser, Y., Kosowski, A., Paj{\k{a}}k, D., Uzna{\'n}ski, P.:
  Fast collaborative graph exploration. Information and Computation
  \textbf{243},  37--49 (2015)

\bibitem{DAFPS12}
Devismes, S., Lamani, A., Petit, F., Raymond, P., Tixeuil, S.: Optimal grid
  exploration by asynchronous oblivious robots. In: Proc. of Symposium on
  Self-Stabilizing Systems (SSS). pp. 64--76 (2012)

\bibitem{DFKP02}
Diks, K., Fraigniaud, P., Kranakis, E., Pelc, A.: Tree exploration with little
  memory. In: Proceedings of the thirteenth annual ACM-SIAM symposium on
  Discrete algorithms. pp. 588--597. Society for Industrial and Applied
  Mathematics (2002)

\bibitem{DKHS06}
Dynia, M., Kuty{\l}owski, J., auf~der Heide, F.M., Schindelhauer, C.: Smart
  robot teams exploring sparse trees. In: International Symposium on
  Mathematical Foundations of Computer Science. pp. 327--338. Springer (2006)

\bibitem{EB11}
Elor, Y., Bruckstein, A.M.: Uniform multi-agent deployment on a ring.
  Theoretical Computer Science  \textbf{412}(8-10),  783--795 (2011)

\bibitem{FIPPP05}
Fraigniaud, P., Ilcinkas, D., Peer, G., Pelc, A., Peleg, D.: Graph exploration
  by a finite automaton. Theor. Comput. Sci.  \textbf{345}(2-3),  331--344
  (2005)

\bibitem{KA18}
Kshemkalyani, A.D., Ali, F.: Efficient dispersion of mobile robots on graphs.
  arXiv preprint arXiv:1805.12242  (2018)

\bibitem{LABMTB10}
Lamani, A., Potop-Butucaru, M., Tixeuil, S.: Optimal deterministic ring
  exploration with oblivious asynchronous robots. Structural Information and
  Communication Complexity pp. 183--196 (2010)

\bibitem{LL93}
Lov\'asz, L.: Random walks on graphs: A survey. Combinatorics, Paul Erd\H{o}s
  is Eighty  \textbf{2},  1--46 (1993)

\bibitem{MGS98}
Muthukrishnan, S., Ghosh, B., Schultz, M.H.: First-and second-order diffusive
  methods for rapid, coarse, distributed load balancing. Theory of computing
  systems  \textbf{31}(4),  331--354 (1998)

\bibitem{OS14}
Ortolf, C., Schindelhauer, C.: A recursive approach to multi-robot exploration
  of trees. In: International Colloquium on Structural Information and
  Communication Complexity. pp. 343--354. Springer (2014)

\bibitem{PV89}
Peleg, D., Van~Gelder, A.: Packet distribution on a ring. Journal of Parallel
  and Distributed Computing  \textbf{6}(3),  558--567 (1989)

\bibitem{R08}
Reingold, O.: Undirected connectivity in log-space. J. {ACM}  \textbf{55}(4),
  17:1--17:24 (2008)

\bibitem{SMOKM16}
Shibata, M., Mega, T., Ooshita, F., Kakugawa, H., Masuzawa, T.: Uniform
  deployment of mobile agents in asynchronous rings. In: Proceedings of the
  2016 ACM Symposium on Principles of Distributed Computing. pp. 415--424. ACM
  (2016)

\bibitem{SS94}
Subramanian, R., Scherson, I.D.: An analysis of diffusive load-balancing. In:
  Proceedings of the sixth annual ACM symposium on Parallel algorithms and
  architectures. pp. 220--225. ACM (1994)

\end{thebibliography}

\end{document}